\documentclass{article}
\usepackage{makeidx}
\usepackage{amssymb}
\usepackage{graphicx}
\usepackage{amsmath}

\newtheorem{theorem}{Theorem}

\newtheorem{definition}[theorem]{Definition}

\newtheorem{lemma}[theorem]{Lemma}

\newtheorem{proposition}[theorem]{Proposition}
\newtheorem{remark}[theorem]{Remark}

\newenvironment{proof}[1][Proof]{\textbf{#1.} }{\ \rule{0.5em}{0.5em}}

\input{tcilatex}
\begin{document}

\title{An elementary approach to Stochastic Differential Equations using the
infinitesimals.}
\author{Vieri Benci\thanks{%
Dipartimento di Matematica Applicata "U.Dini" Via F. Buonarroti 1/c , I -
56127 Pisa \texttt{benci@dma.unipi.it}, \texttt{galatolo@dm.unipi.it},
\texttt{ghimenti@dm.unipi.it} }, Stefano Galatolo\addtocounter{footnote}{-1}%
\footnotemark {}, and Marco Ghimenti\addtocounter{footnote}{-1}\footnotemark %
{} }
\maketitle
\tableofcontents

\section{Introduction}

Suppose that $x$ is a physical quantity whose evolution is governed by a
deterministic force which has small random fluctuations; such a phenomenon
can be described by the following equation
\begin{equation}
\dot{x}=f(x)+h(x)\xi (t)  \label{1}
\end{equation}%
where $\dot{x}=\frac{dx}{dt},$ and $\xi $ is a "white noise". Intuitively, a
\textit{white noise} is the derivative of a \textit{Brownian motion}, namely
a continuous function which is not differentiable in any point.

There is no function $\xi $ which has such a property, actually the mathematical
object which models $\xi $ is a distribution. Thus equation (\ref{1}) makes
sense if it lives in the world of distributions.

On the other hand the kind of problems which an applied mathematician asks
are of the following type. Suppose that $x(0)=0$ and that $\xi (t)$ is a
random noise of which only the statistical properties are known. What is the
probability distribution $P(t,x)$ of $x$ at the time $t?$

This question can be formalized by the theory of stochastic differential
equation and eq. (\ref{1}) takes the form
\begin{equation}
dx=f(x)dt+h(x)dw.  \label{2}
\end{equation}%
thus, the white noise $dw$ is regarded as the "differential" of a Wiener
process $w.$ In this case, both $x(t)$ and $w(t)$ are modelled, not by
distributions, but by stochastic processes.

By the Ito theory, the above question can be solved rigorously: the
probability distribution can be determined solving the Fokker-Plank
equation:
\begin{equation}
\frac{\partial P}{\partial t}=\frac{1}{2}\frac{\partial ^{2}}{\partial x^{2}}%
\left( h(x)^2P\right) +\frac{\partial }{\partial x}\left( f(x)P\right)
\label{FP}
\end{equation}

Eq. (\ref{1}) (or (\ref{2})) and (\ref{FP}) are very relevant in
applications of Mathematics and the practitioners of mathematics such as
engineers, physicists, economists, etc. make a large use of it. However the
mathematics used in these equations is rather involved and many of them are
not able to control it.

Usually people think of some intuitively simpler model. For example, $\xi
(t) $ is considered as a force which acts at discrete instants of time $%
t_{i} $; it is supposed that the difference of two successive times $%
dt=t_{i+1}-t_{i}$ is infinitesimal and its strength is infinite; namely
\begin{equation}
\xi (t)=\pm \frac{1}{\sqrt{dt}}  \label{3}
\end{equation}

The sign of this force is determined by a fair coin tossing. Clearly eq. (%
\ref{3}) makes no sense and the gap between the rigorous mathematical
description and the intuitive model is quite large.

The main purpose of this paper is to reduce this gap. We will use infinite
and infinitesimal numbers in such a way that eqs. (\ref{3}) and (\ref{1})
make sense and, in this framework, we will deduce eq.(\ref{FP}) rigorously.
Our proof is relatively simple and very close to intuition.

The use of infinite and infinitesimal numbers naturally leads to Nonstandard
Analysis (\textit{NSA}). Actually, some mathematicians have described the
stochastic differential equation by Nonstandard Analysis (cf. e.g. \cite{and}%
, \cite{kei}, \cite{albe}, \cite{nel} and references therein). However the
machinery of $NSA$ is too complicate for practitioners of mathematics even
if its ideas are simpler. A good knowledge of formal logic is necessary to
follow a $NSA$ proof, in fact, the main tool is the transfer principle
which, in order to be applied correctly, needs the notions of \textit{formal
language} and \textit{interpretation}.

In this paper we will not use $NSA$ but $\alpha $-theory which is a kind of
simplified version of Nonstandard Analysis. $\alpha $-theory has been
introduced in \cite{bdn} (see also \cite{bb} and \cite{bda}) with the
purpose to provide a simpler approach to $NSA$. In fact, in the quoted paper
it has been proved that a particular model of $NSA$ can be deduced by the
axioms of $\alpha $-theory (we refer also to \cite{bdn2} and to \cite{bfdn}
for the reader interested to investigate in these questions).

The main differences between $\alpha $-theory and the usual Nonstandard
Analysis are two:

\begin{itemize}
\item $\alpha $-theory does not need the language (and the knowledge) of
symbolic logic;

\item it does not need to distinguish two mathematical universes, (the
standard universe and the nonstandard one).
\end{itemize}

$\alpha $-theory postulates the "existence" of an infinite integer number
called $\alpha $ and it provides the rules necessary to deal with the
mathematical objects which can be constructed by its introduction. For
example, $\alpha $-theory allows to define functions such as$\ "\sin \left(
\alpha t\right) "$ and to manage with it. $\alpha $-theory is not as
powerful as $NSA,$ but it is simpler and it allows to treat many problems by
an elementary and rigorous use of infinite and infinitesimal numbers.

In particular, using this theory, it is possible to define the  "grid
functions" which are functions defined for times $t_{i}$ belonging to a set $%
\mathbb{H}$ which models the axis of time. Using the notion of grid
function, we are able to give a sense to (\ref{1}) and (\ref{3}) and to
deduce eq. (\ref{FP}) rigorously.

Our approach presents the following peculiarities:

\begin{itemize}
\item we will rewrite eq. (\ref{1}) as a "grid" differential equation:
\begin{equation}
\frac{\Delta x}{\Delta t}=f(x)+\xi (t)  \label{1+}
\end{equation}

where $\frac{\Delta x}{\Delta t}$ denotes the grid derivative (see Def. \ref%
{dergriglia}). From this equation, it easy to recover both a distribution
equation and a stochastic equation, and, at the same time, eq. (\ref{1+}) has a very
intuitive meaning.

\item when eq. (\ref{1+}) is considered from the stochastic point of view,
the noise $\xi $ is regarded as a \textit{grid function} belonging to the
space of all possible noises $\mathcal{R}.$ If $\xi $ is regarded as a
random variable, the probability on the sample space $\mathcal{R}$ can be
defined in a naive way, namely every noise has the same probability. This is
the basic idea of the Loeb measure (\cite{loeb}) which is an important tool
in the applications of $NSA$ , but we do not need to use it. Actually we do
not need to introduce any kind of measure.
\end{itemize}

\section{The Alpha-Calculus}

\subsection{Basic notions of Alpha-Theory}

In this section, we will expose the basic facts of $\alpha $-theory and the
basic tools which will be used in the paper in a elementary and self
contained way.

$\alpha $-theory is based on the existence of a new mathematical object,
namely $\mathbf{\alpha }$ which is added to the other entities of the
mathematical universe. We may think of $\mathbf{\alpha }$ as a new
\textquotedblleft \textit{ideal}\textquotedblright\ natural number added to $%
{\mathbb{N}}$, in a similar way as the imaginary unit $i$ can be seen as a
new ideal number added to the real numbers ${\mathbb{R}}$. Before going to
the axioms for $\mathbf{\alpha } $, we remark that \textit{all} usual
principles of mathematics are implicitly assumed. Informally, we can say
that, by adopting $\alpha $-theory, one can construct sets and functions
according to the \textquotedblleft usual\textquotedblright\ practice of
mathematics, with no restrictions whatsoever. A precise definition of what
we mean by \textquotedblleft usual principles of
mathematics\textquotedblright\ (i.e. of our underlying set theory) is given
in \cite{bdn}.

Like the use of the imaginary entity $i$ is governed by simple properties
like $i^2=-1$ and the usual rules for the product and sum, the use of $%
\mathbf{\alpha }$ is governed by the following five axioms. \newline

\noindent $\mathbf{\alpha 1}$\textbf{.\ Extension Axiom.} \newline
\textit{Every sequence }$\varphi $ \textit{can be uniquely extended to }${%
\mathbb{N}}\cup \left\{ \mathbf{\alpha }\right\} $\textit{. The
corresponding value at }$\mathbf{\alpha }$ \textit{will be denoted by }$%
\varphi (\mathbf{\alpha )}$\textit{\ and called the} \textit{value of} $%
\varphi $\textit{\ at the point }$\mathbf{\alpha }$ \textit{or more simply}
\textit{\textquotedblleft }$\alpha $-\textit{value\textquotedblright . If
two sequences }$\varphi ,\psi $ \textit{are different at all points, then }$%
\varphi (\mathbf{\alpha })\neq \psi (\mathbf{\alpha })$. \newline

\bigskip

We remark that if $\varphi :{\mathbb{N}}\rightarrow A$, then in general $%
\varphi (\mathbf{\alpha })\notin A$. The \textquotedblleft difference
preserving\textquotedblright\ condition given above can be rephrased as
follows: \textquotedblleft If two sequences are different at all $n$ then
they must be different at the point "$\mathbf{\alpha }$" as
well\textquotedblright . It is a non-triviality condition, that will allow
plenty of values at infinity. Moreover remark that the $\alpha $-value of a
sequence should not be confused with its limit. In fact, the $\alpha $-value
differs from a limit even for this first axiom; in fact different sequences
might have the same limit.

\bigskip

The next axiom gives a natural coherence property with respect to
compositions. If $g:A\rightarrow B$ and $h:B\rightarrow C$, denote by $%
h\circ g:A\rightarrow C$ the composition of $h$ and $g$, i.e. $(h\circ
g)(x)=h(g(x))$. \bigskip

\noindent $\mathbf{\alpha 2}$\textbf{.\ Composition Axiom.} \newline
\textit{If }$\varphi $\textit{\ and }$\psi $\textit{\ are sequences and if }$%
f$ \textit{is any function such that compositions }$f\circ \varphi $\textit{%
\ and }$f\circ \psi $\textit{\ make sense, then} \newline
\begin{equation*}
\varphi (\mathbf{\alpha })=\psi (\mathbf{\alpha })\Rightarrow (f\circ
\varphi )(\mathbf{\alpha })=(f\circ \psi )(\mathbf{\alpha })\mathit{\ }
\end{equation*}%
\newline
So, if two sequences takes the same value at infinity, by composing them
with any function we again get sequences with the same $\alpha$-value.
\bigskip

\noindent $\mathbf{\alpha 3}$\textbf{.\ Real Number Axiom.} \newline
\textit{If }$c_{r}:n\mapsto r$ \textit{is the constant sequence with value }$%
r$\textit{, then }$c_{r}(\mathbf{\alpha })=r$\textit{; if }$1_{{\mathbb{N}}%
}:n\mapsto n$\textit{\ is the immersion of }${\mathbb{N}}$ in $\mathbb{R}$%
\textit{, then }$1_{{\mathbb{N}}}(\mathbf{\alpha })=\mathbf{\alpha }\notin {%
\mathbb{R}}$.

\bigskip

This axiom simply says that, for real numbers, the notions of constant
sequence is preserved at infinity. The latter condition says that the ideal
number $\mathbf{\alpha }$ is actually a \textit{new} number. Thus the
immersion $1_{{\mathbb{N}}}$ provides a first example of sequence $\varphi :{%
\mathbb{N}\rightarrow \mathbb{N}}$ such that $\varphi (\mathbf{\alpha }%
)\notin {\mathbb{N}}$.

\bigskip

\noindent $\mathbf{\alpha 4}$\textbf{.\ Internal Set Axiom.} \newline
\textit{If }$\psi $\textit{\ is a sequence of sets,} \textit{then also }$%
\psi (\mathbf{\alpha })$\textit{\ is a set and}%
\begin{equation*}
\psi (\mathbf{\alpha })=\left\{ \varphi (\mathbf{\alpha }):\varphi (n)\in
\psi (n)\ \text{\textit{for all}}\ n\right\} .
\end{equation*}

Thus, the membership relation is preserved at infinity. That is, if $\varphi
(n)\in \psi (n)$ for all $n$, then $\varphi (\mathbf{\alpha })\in \psi (%
\mathbf{\alpha })$. Besides, all elements of $\psi (\mathbf{\alpha })$ are
obtained in this way. That is, they all are values at infinity of sequences
which are pointwise members of $\psi $. The set considered above will be
called \emph{Internal sets}.

\bigskip

\noindent $\mathbf{\alpha 5}$\textbf{.\ Pair Axiom.} \newline
\textit{If }$\vartheta (n)=$\textit{\ }$\left\{ \varphi (n),\psi (n)\right\}
$\textit{\ for all }$n$, \textit{then }$\vartheta (\mathbf{\alpha })=$%
\textit{\ }$\left\{ \varphi (\mathbf{\alpha }),\psi (\mathbf{\alpha }%
)\right\} $\textit{.}

\bigskip

Thus, if the sequence $\xi $\ is such that either $\xi (n)=$\ $\varphi (n)$\
or $\xi (n)=$\ $\psi (n)$\ for all $n$,\ then either $\xi (\mathbf{\alpha }%
)= $\ $\psi (\mathbf{\alpha })$\ or $\xi (\mathbf{\alpha })=$\ $\psi (%
\mathbf{\alpha })$ at infinity as well. As a straight consequence of the
last two axioms, any constant sequence with value a finite set of natural
numbers, or a finite set of finite sets of natural numbers etc., takes the
same value at infinity as well. We remark that this is not true in general.

\

We remark that the above five axioms are given somewhat \textquotedblleft
informally\textquotedblright . Precise indications for a rigorous
formulation as sentences of a suitable first-order language are given in
\cite{bdn}. Also, we refer to \cite{bdn} for the proves of the propositions
below, but we suggest the reader to try them by himself to get acquainted
with $\alpha $-theory.

\begin{definition}
\label{one}If $A$ is a set, the $\ast $-transform of $A$ is defined as
follows:

\begin{equation*}
A^{\ast }=\{\varphi (\mathbf{\alpha }):\varphi :{\mathbb{N}}\rightarrow A\}.
\end{equation*}
\end{definition}

If $\psi \ $is a sequence such that $\psi (n)=A\ $\textit{for all}$\ n,$then
by the Internal Set Axiom, we have that $\psi (\mathbf{\alpha })=A^{\ast }.$
Then constant set-valued sequences behave differently than real valued
sequences (cf. the Real Number Axiom).

\begin{definition}
The set of the \textit{hyperreal numbers} is the $\ast $-transform of the
set of the real numbers:%
\begin{equation*}
{\mathbb{R}}^{\ast }=\{\varphi (\mathbf{\alpha }):\varphi :{\mathbb{N}}%
\rightarrow {\mathbb{R}}\}.
\end{equation*}
\end{definition}

In other words, the hyperreal numbers are the $\alpha $-values assumed
by real sequences. With obvious notation, for instance we will write $\sin
\frac{2}{\mathbf{\alpha }}$ to mean the hyperreal number obtained as the
values at infinity of the sequence $\left\{ \sin \frac{2}{n}\right\} _{n\in
{\mathbb{N}}}$.

The sum and product operation are naturally transported on the hyperreal set
moreover we have the following:

\begin{proposition}
\textit{The hyperreal number system }$\left\langle {\mathbb{R}}^{\ast
};+,\cdot ,0,1,<\right\rangle $ \textit{is an ordered field.}
\end{proposition}

Besides the considered sets of hyper-numbers, another fundamental notion in
nonstandard analysis is the following.

\begin{definition}
A set $\Gamma \subset A^{\ast }$ is called \textit{hyperfinite }if
\begin{equation*}
\Gamma =\left\{ \varphi (\alpha ):\varphi (n)\in A_{n}\right\}
\end{equation*}%
where\textit{\ }$A_{n}\subset A$\textit{\ }is a sequence of finite sets.
Given a hyperfinite set $\Gamma $, we define its cardinality $\left\vert
\Gamma \right\vert $ as follows:%
\begin{equation*}
\left\vert \Gamma \right\vert =\psi \left( \alpha \right) \in \mathbb{N}%
^{\ast }
\end{equation*}%
where $\psi \left( n\right) =\left\vert A_{n}\right\vert $ is the
cardinality of the finite set $A_{n}.$
\end{definition}

In general hyperfinite sets are infinite; their importance relies in the
fact that they retain all \textquotedblleft elementary\textquotedblright\
properties of finite sets. Applications of hyperfinite sets will be given in
subsequent sections, for example the following hold

\begin{proposition}
Every nonempty hyperfinite subset of ${\mathbb{R}}^{\ast }$\textit{\ has a
greatest and a smallest element.}
\end{proposition}

A very important example of hyperfinite set which we will use in this paper
is the \textit{hyperfinite grid }${\mathbb{H}}$. The \textit{hyperfinite
grid }${\mathbb{H}}_{\alpha }$ is defined as the ${\alpha }$-value of the
set
\begin{equation*}
{\mathbb{H}}_{n}=\left\{ \frac{k}{n}:k\in \mathbb{Z},\ -\frac{n^{2}}{2}\leq
k<\frac{n^{2}}{2}\right\} ;
\end{equation*}%
namely,
\begin{equation*}
{\mathbb{H}}_{\alpha }:=\left\{ \frac k \alpha :k\in \mathbb{Z}^{\ast },\ -%
\frac{\alpha ^{2}}{2}\leq k<\frac{\alpha ^{2}}{2}\right\}
\end{equation*}%
In the follwing, for short, usually we will write ${\mathbb{H}}$ instead of $%
{\mathbb{H}}_{\alpha }$. Clearly ${\mathbb{H}}$ is an hyperfinite set with $%
\left\vert {\mathbb{H}}\right\vert =\alpha ^{2}$. Given $a,b\in {\mathbb{H}}%
, $ we set%
\begin{eqnarray*}
\lbrack a,b]_{\mathbb{H}} &\mathbf{=}&\left\{ x\in \mathbb{H}:a\leq k\leq
b\right\} \\
\lbrack a,b)_{\mathbb{H}} &\mathbf{=}&\left\{ x\in \mathbb{H}:a\leq
k<b\right\}
\end{eqnarray*}%
\bigskip

If we identify the functions with their graphs, $f^{\ast }$ is defined by
 definition \ref{one} and it is not difficult to prove the following

\begin{proposition}
\textit{Let }$f:A\rightarrow B$\textit{\ be a function. Then its
star-transform }$f^{\ast }$\textit{\ is a function }$f^{\ast }:A^{\ast
}\rightarrow B^{\ast }$ \textit{and, for every sequence }$\varphi :{\mathbb{N%
}}\rightarrow A$,%
\begin{equation*}
f^{\ast }(\varphi (\mathbf{\alpha }))=(f\circ \varphi )^{\ast }(\mathbf{%
\alpha })
\end{equation*}%
\noindent \textit{Moreover, }$f^{\ast }$\textit{\ is 1-1 (or onto) iff }$f$
\textit{is 1-1 (or onto, respectively).}
\end{proposition}

When confusion is unlikely, we will omit the symbol "$*$" and "$f^{\ast }$"
will be denoted by "$f$".

Let $f_{n}:A\rightarrow B$ be a sequence of functions; then identifying the
functions with their graphs $f_{\alpha }$ is well defined by axiom $\left(
\alpha 3\right) $ and we have that%
\begin{equation*}
f_{\alpha }:A^{\ast }\rightarrow B^{\ast }
\end{equation*}%
is a function defined by%
\begin{equation*}
f_{\alpha }(\varphi (\alpha ))=\psi \left( \alpha \right)
\end{equation*}%
where $\psi \left( n\right) :=f_{n}(\varphi (n))$ is a sequence in $B.$

\begin{definition}
A function
\begin{equation*}
f:A^{\ast }\rightarrow B^{\ast }
\end{equation*}%
is called internal if it is the graph of an internal set, namely if there is
a sequence of functions $f_{n}:A\rightarrow B$ such that%
\begin{equation*}
f=f_{\alpha }
\end{equation*}
\end{definition}

\subsection{Infinitesimally small and infinitely large numbers.}

A fundamental feature of $\alpha $-calculus is that the intuitive notions of
\textquotedblleft\ infinitesimally small\textquotedblright\ number and
\textquotedblleft\ infinitely large\textquotedblright\ number can be
formalized as actual objects of the hyperreal line. This give many
possibilities to simplify proofs and statements in calculus theory.

\begin{definition}
A hyperreal number $\xi \in {\mathbb{R}}^{\ast }$ is \textit{bounded} or
\textit{finite} if its absolute value $\left| \xi \right| <r $ for some $%
r\in {\mathbb{R}}$. We say that $\xi $ is \textit{unbounded} or \textit{%
infinite} if it is not bounded. $\xi $ is \textit{infinitesimal} if $|\xi
|<r $ for all positive $r\in {\mathbb{R}}$.
\end{definition}

Clearly, the inverse of an infinite number is infinitesimal and vice versa,
i.e. the inverse of a (nonzero) infinitesimal number is infinite. An example
of an infinitesimal is given by $\bigcirc\hspace{-6.5pt}\varepsilon\
:=1/\alpha $, the $\alpha $ -value of the sequence $\left\{ 1/n\right\} $.

From now on, the symbol $\bigcirc\hspace{-6.5pt}\varepsilon\ $ will always
denote $1/\alpha .$

All infinitesimal and all real numbers are bounded. However there are finite
hyperreals that are neither infinitesimal nor real, for example $5+\bigcirc
\hspace{-6.5pt}\varepsilon \ $ and $7+\sin \mathbf{\alpha }$.

\begin{definition}
We say that two hyperreal numbers $\xi $ and $\eta $ are \textit{%
infinitesimally close} if $\xi -\eta $ is infinitesimal. In this case we
write $\xi \sim \eta $.
\end{definition}

It is easily seen that $\sim $ is an equivalence relation.

On the other hand (as it is intuitive) each bounded hyperreal is infinitely
close to some real. The following indeed comes from the completeness of the
real line.

\begin{theorem}[Shadow Theorem]
\textit{Every bounded hyperreal number }$\xi $\textit{\ is infinitesimally
close to a unique real number }$r$\textit{, called the shadow of }$\xi $%
\textit{. Symbolically }$r=sh(\xi )$\textit{. }
\end{theorem}

The notion of a shadow is extended to every hyperreal number, by setting $%
sh(\xi )=+\infty $ if $\xi $ is positive unbounded, and $sh(\xi )=-\infty $
if $\xi $ is negative unbounded.

\begin{definition}
Given two hyperreal numbers $\xi $ and $\zeta \in {\mathbb{R}}^{\ast
}\backslash \left\{ 0\right\} ,$ we say that they have the same order if $%
\xi /\zeta $ and $\zeta /\xi $ are bounded numbers and we will write
\begin{equation*}
\xi \approx \zeta
\end{equation*}%
(notice the difference between $"\sim "$ and $"\approx "$ since these
symbols will be largely used in the rest of this paper). We say that $\xi $
has a larger order than $\zeta $ if $\xi /\zeta $ is an infinite number and
we will write%
\begin{equation*}
\xi \gg \zeta
\end{equation*}
We say that $\xi $ has a smaller order than $\zeta $ if $\xi /\zeta $ is an
infinitesimal number and we will write%
\begin{equation*}
\xi \ll \zeta
\end{equation*}
\end{definition}

\subsection{Some notions of infinitesimal calculus}

Now we see how all this machinery can be used to build a rigorous
``infinitesimal'' calculus. We present how the definition of limit can be
given in our setting.

\begin{definition}
We say that $\lim_{x\rightarrow x_{0}}f(x)=l$ if $f^{\ast }(\xi )\sim l$ for
all $\xi \sim x_{0}$ ($\xi \neq x_{0}$).
\end{definition}

With the definition of limit all the elementary calculus can be
reconstructed, but the features of our method allow to avoid the use of
limits and work with \emph{real} infinitesimal and infinite numbers. Let us
see some example: the definition of continuity and derivative. We remark
that the theory given by these definitions is equivalent to the standard
calculus and all the known results (as for example the Lagrange's or
Fermat's theorems) applies.

\begin{definition}
\ A real function $f:A\rightarrow {\mathbb{R}}$ is \textit{continuous} at $%
x_{0}\in A$ if for every $\xi \in A^{\ast }$, $\xi \sim x_{0}\Rightarrow
f^*(\xi )\sim f^*(x_{0})$.
\end{definition}

Let $f$ be any real function defined on a neighborhood of $x_{0}$.

\begin{definition}
We say that $f$ has derivative at $x_{0}$ if there exists $f^{\prime
}(x_{0})\in {\mathbb{R}}$ such that for all infinitesimals $\varepsilon \neq
0$,

\begin{equation*}
\frac{f^{\ast }(x_{0}+\varepsilon )-f^{\ast }(x_{0})}{\varepsilon }\sim
f^{\prime }(x_{0})
\end{equation*}

Equivalently, $f$ has \textit{derivative} $f^{\prime }(x_{0})$ at $x_{0}$ if
for every infinitesimal $\varepsilon $ there is an infinitesimal $\delta $
such that $f^{\ast }(x_{0}+\varepsilon )=f(x_{0})+f^{\prime
}(x_{0})\varepsilon +\delta \varepsilon $.
\end{definition}

As said before all the classical results of calculus hold in this
framework. An example which will be used in the following is the Taylor
formula (with infinitesimal remainder).

\begin{theorem}
\label{taylor} If $f\in C^{n+1}(\mathbb{R})$ then for each infinitesimal $%
\epsilon $ there is an infinitesimal $\eta $ such that
\begin{equation*}
f^{\ast }(x+\epsilon )=\sum_{k\leq n}\frac{f^{(k)}(x)\epsilon ^{k}}{k!}+\eta
\epsilon ^{n}.
\end{equation*}
\end{theorem}

Now we introduce a concept of integral. This concept is more general than
the Riemann integral and will allow us to integrate noises and stochastic
equations. Intuitively this integral is just an infinite sum of hyperreal
numbers. This sum will be done on an hyperfinite set.

\begin{definition}
If $\Gamma =\chi (\mathbf{\alpha })$ is a hyperfinite set of hyperreal
numbers, then its \textit{hyperfinite sum: }%
\begin{equation*}
\sum_{x\in \Gamma }x=Sum_{\chi }(\mathbf{\alpha })
\end{equation*}
is defined as the value at infinity of the sequence of finite sums
\begin{equation*}
Sum_{\chi }(n)=\sum_{x\in \chi (n)}x.
\end{equation*}
\end{definition}

It is easily checked that this definition does not depend on the choice of
the sequence $\left\{ \chi (n)\right\} $, but only on its value at infinity $%
\Gamma $. Using this definition, we define the $\alpha $-$integral$.

\begin{definition}
\label{17}Let $f:A\rightarrow {\mathbb{R}}$ be \textit{any} function, where $%
A \subseteq \mathbb{R}$. Its \textit{Alpha-integral} on $A$, denoted by $%
\int_{A}f(x)\,\Delta x$, is the number in $\mathbb{R} \cup \left\{ {\mathbb{%
\pm }\infty }\right\} $ defined as the shadow of the following hyperfinite
sum:

\begin{equation*}
\int_{A}f(x)\,\,\Delta x\ =\ sh\left( \bigcirc\hspace{-6.5pt}\varepsilon\
\cdot \sum_{\xi \in {\mathbb{H}\cap }A^{\ast }}f^{\ast }(\xi )\right) \
\end{equation*}
\end{definition}

Notice that

\begin{equation*}
\int_{A}f(x)\,\,\Delta x\ =\ sh\left( S_{A}(\alpha \mathbf{)}\right) \ \text{%
where}\ S_{A}(n\mathbf{)=}\frac{1}{n}\cdot \sum_{x\in {\mathbb{H}}(n){%
\mathbb{\cap }}A}f^{\ast }(x)\
\end{equation*}

Of course, if $A=\left[ a,b\right] $ is a closed interval, we adopt the
usual notation $\int_{a}^{b}f(x)\,\,\Delta x$.

The Alpha-integral $\int_{a}^{b}f(x)\,\,\Delta x$ is defined for every
function. In fact, while the sequence
\begin{equation*}
S_{a}^{b}(n)=\frac{1}{n}\cdot \sum_{x\in {\mathbb{H}}(n){\mathbb{\cap }}%
(a,b)}f^{\ast }(x)
\end{equation*}
may not have a limit in the classic sense, its $\alpha $-value\ $%
S_{a}^{b}(\alpha )$ is always defined. If the function $f$ is Riemann
integrable then $\lim_{n\rightarrow \infty }S_{a}^{b}(n)$ exists and
coincides with the $\alpha $-integral (notice that if a\ real sequence $%
\left\{ \varphi (n)\right\} $ has \textquotedblleft
classic\textquotedblright\ limit $l\in {\mathbb{R}\cup }\left\{ \pm \infty
\right\} $, then it must be $sh\left( \varphi (\alpha )\right) =l$). Thus
the Alpha-integral actually generalizes Riemann integral.

\section{Grid functions}

A grid function is a function whose argument range on an hyperfinite "grid"
whose elements are the (hypernatural) multiples of $\frac{1}{\alpha }.$
Since the grid is hyperfinite these functions are easy to handle and from
many points of view they behave similarly to functions on finite sets. We
will see that this simple kind of functions are flexible enough to contain
elements representing distributions. This flexibility will allow us to
obtain in a simple way a kind of stochastic calculus ( see, e.g. the Ito's
formula, Thm. \ref{ITO}).

\subsection{Basic notions}

An internal function%
\begin{equation*}
\xi :{\mathbb{H}}\rightarrow {\mathbb{R}}^{\ast }
\end{equation*}%
is called \textit{grid function}.

\begin{definition}
\label{dergriglia}Given a grid function $\xi :{\mathbb{H}}\rightarrow {%
\mathbb{R}}^{\ast }$, we define its grid derivative $\frac{\Delta \xi }{%
\Delta t}$ as
\begin{equation*}
\frac{\Delta \xi }{\Delta t}(t)=\frac{\xi (t+\bigcirc\hspace{-6.5pt}%
\varepsilon\ )-\xi (t)}{\bigcirc\hspace{-6.5pt}\varepsilon\ };
\end{equation*}%
The grid integral of $\xi $ is defined as
\begin{equation*}
\mathbb{I}\left[ \xi \right] =\bigcirc\hspace{-6.5pt}\varepsilon\ \sum_{t\in
\mathbb{H}}\xi \left( t\right) ;
\end{equation*}%
if $\Gamma \subset \mathbb{H}$ is a hyperfinite set\ we define $\mathbb{I}_{{%
\Gamma }}\left[ \xi \right] $, its grid integral in ${\Gamma }$, as
\begin{equation*}
\mathbb{I}_{{\Gamma }}\left[ \xi \right] =\bigcirc\hspace{-6.5pt}%
\varepsilon\ \sum_{t\in {\Gamma }}\xi \left( t\right)
\end{equation*}
\end{definition}

Most of the properties of the usual derivative hold also for the grid
derivative, for example we have that, if $\xi $ and $\zeta $ are continuous
functions, with finite grid derivative,%
\begin{eqnarray*}
\frac{\Delta (\xi \zeta )}{\Delta t} &=&\frac{\xi (t+\bigcirc\hspace{-6.5pt}%
\varepsilon\ )\zeta (t+\bigcirc\hspace{-6.5pt}\varepsilon\ )-\xi (t)\zeta (t)%
}{\bigcirc\hspace{-6.5pt}\varepsilon\ }= \\
&=&\frac{\xi (t+\bigcirc\hspace{-6.5pt}\varepsilon\ )\zeta (t+\bigcirc%
\hspace{-6.5pt}\varepsilon\ )-\xi (t+\bigcirc\hspace{-6.5pt}\varepsilon\
)\zeta (t)+\xi (t+\bigcirc\hspace{-6.5pt}\varepsilon\ )\zeta (t)-\xi
(t)\zeta (t)}{\bigcirc\hspace{-6.5pt}\varepsilon\ }= \\
&=&\frac{\Delta \xi }{\Delta t}(t)\cdot \zeta (t)+\xi (t+\bigcirc\hspace{%
-6.5pt}\varepsilon\ )\cdot \frac{\Delta \zeta }{\Delta t}(t)\sim \frac{%
\Delta \xi }{\Delta t}\cdot \zeta +\xi \cdot \frac{\Delta \zeta }{\Delta t}.
\end{eqnarray*}

These notions can be easily extended to functions of more variables; for
example if
\begin{equation*}
\rho (t,x):{\mathbb{H}}\times {\mathbb{H}}\rightarrow {\mathbb{R}}^{\ast }
\end{equation*}%
we set%
\begin{eqnarray*}
\frac{\Delta \rho }{\Delta t}(t,x) &=&\frac{\rho (t+\bigcirc\hspace{-6.5pt}%
\varepsilon\ ,x)-\rho (t,x)}{\bigcirc\hspace{-6.5pt}\varepsilon\ } \\
\frac{\Delta \rho }{\Delta x}(t,x) &=&\frac{\rho (t,x+\bigcirc\hspace{-6.5pt}%
\varepsilon\ )-\rho (t,x)}{\bigcirc\hspace{-6.5pt}\varepsilon\ }
\end{eqnarray*}%
and if ${\Gamma }\subset \mathbb{H}^{2}$ is a hyperfinite set\ we define its
grid integral $\mathbb{I}_{{\Gamma }}\left[ \rho \right] $ as%
\begin{equation*}
\mathbb{I}_{{\Gamma }}\left[ \rho \right] =\bigcirc\hspace{-6.5pt}%
\varepsilon\ ^{2}\sum_{(t,x)\in {\Gamma }}\rho \left( t,x\right) .
\end{equation*}%
It is clear that the derivative of a grid function $\xi $ is a grid
function. Moreover, if $\xi $ is a grid function, then the \textit{\textit{%
grid integral function }}$x\mapsto \mathbb{I}_{\left[ a,x\right) }\left[ \xi %
\right] $ is a grid function. We have the following relation between the
grid-derivative and the grid-integral:

\begin{theorem}
\label{teofondcalc}If $\xi $ is a grid function, then
\begin{eqnarray*}
\mathbb{I}_{\left[ x,y\right) }\left[ \frac{\Delta \xi }{\Delta x}\right]
&=&\xi \left( y\right) -\xi \left( x\right) \\
\frac{\Delta }{\Delta x}\mathbb{I}_{[a,x)}\left[ \xi \right] &=&\xi \left(
x\right)
\end{eqnarray*}

\begin{proof}
Obviously we have%
\begin{eqnarray*}
\mathbb{I}_{\left[ x,y\right) }\left[ \frac{\Delta \xi }{\Delta x}\right]
&=&\bigcirc\hspace{-6.5pt}\varepsilon\ \sum \frac{\xi (x+\bigcirc\hspace{%
-6.5pt}\varepsilon\ )-\xi (x)}{\bigcirc\hspace{-6.5pt}\varepsilon\ }+\frac{%
\xi (x+2\cdot \bigcirc\hspace{-6.5pt}\varepsilon\ )-\xi (x+\bigcirc\hspace{%
-6.5pt}\varepsilon\ )}{\bigcirc\hspace{-6.5pt}\varepsilon\ }\dots +\frac{\xi
(y)-\xi (y-\bigcirc\hspace{-6.5pt}\varepsilon\ )}{\bigcirc\hspace{-6.5pt}%
\varepsilon\ }= \\
&=&\xi (y)-\xi (x).
\end{eqnarray*}%
Furthermore%
\begin{eqnarray*}
\frac{\Delta }{\Delta x}\mathbb{I}_{[a,x)}\left[ \xi \right] &=&\frac{%
\mathbb{I}_{[a,x+\bigcirc\hspace{-6.5pt}\varepsilon\ )}\left[ \xi \right] -%
\mathbb{I}_{[a,x+\bigcirc\hspace{-6.5pt}\varepsilon\ )}\left[ \xi \right] }{%
\bigcirc\hspace{-6.5pt}\varepsilon\ }= \\
&=&\sum_{t\in \lbrack a,x+\bigcirc\hspace{-6.5pt}\varepsilon\ )}\xi
(t)-\sum_{t\in \lbrack a,x)}\xi (t)=\xi (x).
\end{eqnarray*}
\end{proof}
\end{theorem}

\begin{definition}
\label{int}of A grid function $\xi $ is called integrable in $\left[ a,b%
\right] $ if $\mathbb{I}_{\left[ a,b\right] }\left[ \xi \right] $ is finite;
in this case, we set
\begin{equation*}
\int_{a}^{b}\xi (s)\,\Delta s:=\ sh\left( \mathbb{I}_{[a,b)}\left[ \xi %
\right] \right) =sh\left( \bigcirc\hspace{-6.5pt}\varepsilon\ \sum_{t\in {%
\mathbb{H}\cap \lbrack }a,b)}\xi (t)\right)
\end{equation*}%
$\xi $ is called absolutely integrable in $\left[ a,b\right] $ if $\mathbb{I}%
_{[a,b)}\left[ \left\vert \xi \right\vert \right] $ is finite. $%
\int_{a}^{b}\xi (s)\,ds$ will be called $\alpha $-integral of $\xi $.
\end{definition}

Of course, this integral is strictly related to the $\alpha $-integral given
in Def. \ref{17}. In fact, to every real function
\begin{equation*}
f:\left[ a,b\right] \rightarrow \mathbb{R}
\end{equation*}%
it is possible to associate its natural extension
\begin{equation*}
f^{\mathbb{\ast }}:\left[ a,b\right] ^{\mathbb{\ast }}\rightarrow \mathbb{R}%
^{\mathbb{\ast }}
\end{equation*}%
and a grid function%
\begin{equation}
\tilde{f}:\left[ a,b\right] _{{\mathbb{H}}}\rightarrow \mathbb{R}^{\mathbb{%
\ast }}  \label{gri}
\end{equation}%
obtained as restriction of $f^{\mathbb{\ast }}$ to $\left[ a,b\right] _{{%
\mathbb{H}}}.$ When no ambiguity is possible we will denote $f^{\mathbb{\ast
}} $ and $\tilde{f}$ with the same symbol.

The $\alpha $-integral of $f$ coincides with the $\alpha $-integral of $%
\tilde{f}$ given by Def. \ref{int}.

\subsection{The Ito formula}

We show the power of the grid functions approach by stating in a very simple
way a proposition which is, in some sense, a variant of the Ito's formula.
As in the standard approach this formula will be the main tool in the study
of grid stochastic equations.

\begin{theorem}[Nonstandard Ito's Formula]
\label{ITO} Let $\varphi \in C_{0}^{3}(\mathbb{R}^{2})$ and $x(t)$ be a grid
function such that
\begin{equation}
\left\vert \frac{\Delta x}{\Delta t}(t)\right\vert \leq \eta \alpha ^{2/3},
\label{b}
\end{equation}%
where $\eta \sim 0$.

Then
\begin{equation*}
\frac{\Delta }{\Delta t}\varphi (t,x(t))\sim \varphi _{t}(t,x(t))+\varphi
_{x}(t,x(t))\frac{\Delta x}{\Delta t}(t)+\frac{\bigcirc\hspace{-6.5pt}%
\varepsilon\ }{2}\varphi _{xx}(t,x(t))\cdot \left( \frac{\Delta x}{\Delta t}%
(t)\right) ^{2}.
\end{equation*}%
Here $\varphi _{t},$ $\varphi _{x}$ and $\varphi _{xx}$ denote the usual
partial derivative of $\varphi .$
\end{theorem}

\begin{proof}
By definition of grid derivative we have that%
\begin{eqnarray*}
\frac{\Delta }{\Delta t}\varphi (t,x(t)) &=&\frac{\varphi (t+\bigcirc\hspace{%
-6.5pt}\varepsilon\ ,x(t+\bigcirc\hspace{-6.5pt}\varepsilon\ ))-\varphi
(t,x(t+{\bigcirc\hspace{-6.5pt}\varepsilon\ }))}{{\bigcirc\hspace{-6.5pt}%
\varepsilon\ }}+\frac{\varphi (t,x(t+{\bigcirc\hspace{-6.5pt}\varepsilon\ }%
))-\varphi (t,x(t))}{\bigcirc\hspace{-6.5pt}\varepsilon\ } \\
&\sim &\varphi _{t}(t,x(t+{\bigcirc\hspace{-6.5pt}\varepsilon\ }))+\frac{%
\varphi (t,x(t+{\bigcirc\hspace{-6.5pt}\varepsilon\ }))-\varphi (t,x(t))}{%
\bigcirc\hspace{-6.5pt}\varepsilon\ } \\
&\sim &\varphi _{t}(t,x(t))+\frac{\varphi (t,x(t+{\bigcirc\hspace{-6.5pt}%
\varepsilon\ }))-\varphi (t,x(t))}{\bigcirc\hspace{-6.5pt}\varepsilon\ }
\end{eqnarray*}

But
\begin{equation*}
\varphi (t,x(t+\bigcirc\hspace{-6.5pt}\varepsilon\ ))=\varphi \left(
t,x(t)+\bigcirc\hspace{-6.5pt}\varepsilon\ \frac{\Delta x}{\Delta t}%
(t)\right) ,
\end{equation*}%
and $\left\vert \bigcirc\hspace{-6.5pt}\varepsilon\ \frac{\Delta x}{\Delta t}%
(t)\right\vert \leq \eta \alpha ^{2/3}\cdot\bigcirc\hspace{-6.5pt}%
\varepsilon\ =\eta\cdot \bigcirc\hspace{-6.5pt}\varepsilon\ ^{1/3}$ is
infinitesimal. Then, using the Taylor formula (Theorem \ref{taylor}), we
have that
\begin{eqnarray*}
\varphi \left( t,x(t)+\bigcirc\hspace{-6.5pt}\varepsilon\ \frac{\Delta x}{%
\Delta t}(t)\right) &=&\varphi (t,x(t))+\varphi _{x}(t,x(t))\cdot\bigcirc%
\hspace{-6.5pt}\varepsilon\ \frac{\Delta x}{\Delta t}(t) \\
&& +\frac{1}{2} \varphi _{xx}(t,x(t))\left( \bigcirc\hspace{-6.5pt}%
\varepsilon\ \frac{\Delta x}{\Delta t}(t)\right) ^{2} \\
&&+\frac{1}{3!}\varphi _{xxx}(t,x(t))\left( \bigcirc\hspace{-6.5pt}%
\varepsilon\ \frac{\Delta x}{\Delta t}(t)\right) ^{3}+\delta\left( \bigcirc%
\hspace{-6.5pt}\varepsilon\ \frac{\Delta x}{\Delta t}(t)\right) ^{3}
\end{eqnarray*}%
where $\delta$ is an infinitesimal; hence
\begin{eqnarray*}
\frac{\varphi (t,x(t+\bigcirc\hspace{-6.5pt}\varepsilon\ ))-\varphi (t,x(t))%
}{\bigcirc\hspace{-6.5pt}\varepsilon\ } &=&\varphi _{x}(t,x(t))\frac{\Delta x%
}{\Delta t}(t)+\frac{\bigcirc\hspace{-6.5pt}\varepsilon\ }{2}\varphi
_{xx}(t,x(t))\cdot \left( \frac{\Delta x}{\Delta t}(t)\right) ^{2} \\
&&+\frac{\bigcirc\hspace{-6.5pt}\varepsilon\ ^{2}}{6}\varphi _{xxx}\cdot
\left( \frac{\Delta x}{\Delta t}(t)\right) ^{3}+\delta\cdot \bigcirc\hspace{%
-6.5pt}\varepsilon\ ^{2}\left( \frac{\Delta x}{\Delta t}(t)\right) ^{3}
\end{eqnarray*}%
By the assumption (\ref{b}) the last two terms are infinitesimal and we get
the required result.
\end{proof}

\subsection{Distributions and grid functions}

The grid functions can be considered as a sort of generalization of the
usual real functions.

In fact to every real function correspond a unique grid functions given by (%
\ref{gri}). In the traditional analysis the most important generalization of
the real function is given by the distribution. In this section we will show
that the grid functions represent also a generalization of the notion of
\textit{distribution}.

First of all we recall some notation: given a set $A\subset \mathbb{R}^{N},$
$\mathcal{D}\left( A\right) $ denotes the space of $C^{\infty }$ functions
with compact support of $A.$ The space of the distributions $\mathcal{D}%
^{\prime }\left( A\right) $ is the topological dual of $\mathcal{D}\left(
A\right) $ when $\mathcal{D}\left( A\right) $ is equipped with the Schwartz
topology.

Actually, $\mathcal{D}^{\prime }\left( A\right) $ can also be constructed
without knowing the Schwartz topology by using the notion of grid function.
Next, we will show how to do it.

Let $\mathfrak{G}\left( A\right) $ denote the set of grid function defined
on
\begin{equation*}
A_{\mathbb{H}}:=A^{\ast }\cap \mathbb{H}^{N}
\end{equation*}
On $\mathfrak{G}\left( A\right) $ we define the following equivalence
relation:

\begin{definition}
\label{equiv} Two grid functions $\xi _{1},$ $\xi _{2}$ are said to be
equivalent if
\begin{equation*}
\forall \varphi \in \mathcal{D},\ \ \int \left( \xi _{1}-\xi _{2}\right)
\varphi ds=0
\end{equation*}%
In this case we will write%
\begin{equation*}
\xi _{1}\sim _{\mathcal{D}}\xi _{2}
\end{equation*}
\end{definition}

We may think that two grid functions are equivalent if they are \textit{%
macroscopically equal}.

Moreover, we set
\begin{equation*}
\mathfrak{G}_{0}\left( A\right) =\left\{ \xi \in \mathfrak{G}\left( A\right)
:\forall \varphi \in \mathcal{D},\ \mathbb{I}_{A_{\mathbb{H}}}\left[ \xi
\varphi \right] \ \ is\ finite\right\}
\end{equation*}%
The set of distributions $\mathcal{D}^{\prime }\left( A\right) $ can be
defined as follows%
\begin{equation*}
\mathcal{D}^{\prime }\left( A\right) =\frac{\mathfrak{G}_{0}\left( A\right)
}{\sim _{\mathcal{D}}}.
\end{equation*}%
Thus a distribution can be considered as an equivalence class $T_{\xi }$ of
some grid function $\xi \in \mathfrak{G}_{0}\left( A\right) .$

$T_{\xi }$ can be identified with an element of $\mathcal{D}^{\prime }\left(
A\right) $ by the following formula:%
\begin{equation}
\left\langle T_{\xi },\varphi \right\rangle =\int_{A}\xi \varphi
\,ds=sh\left( \bigcirc\hspace{-6.5pt}\varepsilon\ \cdot \sum_{t\in A_{%
\mathbb{H}}}\xi (t)\varphi ^{\ast }(t)\right) ,\ \ \varphi \in \mathcal{D}.
\label{d}
\end{equation}

To each distribution we can associate a grid function. For example, if $T\in
\mathcal{D}^{\prime }\left( \mathbb{R}\right) $ we can do in the following
way. Since a distribution $T$ has the following representation\footnote{%
See Rudin, functional analysis, Th. 6.28, pag.169}:
\begin{equation*}
T=\sum_{k=0}^{\infty }D^{k}f_{k}
\end{equation*}%
where $f_{k}\;$are continuous. Then the grid function $\xi $ corresponding
to $T$ is given by
\begin{equation*}
\xi \left( t\right) =\sum_{k=0}^{\alpha }\frac{\Delta ^{k}}{\Delta t^{k}}%
f_{k}\left( t\right) .
\end{equation*}

Let us see some simple example. The function
\begin{equation*}
\delta (t)=\alpha \delta _{0,t}
\end{equation*}%
where $\delta _{i,j}$ is the Kronecker symbol correspond to the Dirac $%
\delta $. But also the following grid functions
\begin{equation*}
\alpha \frac{\delta _{0,t}+\delta _{\Delta ,t}}{2};\;\sum \delta _{0,t+\frac{%
k}{\alpha }}\;(k\in \mathbf{Z});\;\text{ etc}
\end{equation*}%
correspond to the Dirac $\delta $. The grid function
\begin{equation*}
\frac{\Delta \delta }{\Delta t}(t)=\alpha ^{2}\left( \delta _{0,t}-\delta
_{\bigcirc\hspace{-5.5pt}\varepsilon\ ,t}\right)
\end{equation*}%
correspond to $\delta^{\prime }.$

The grid function $\alpha ^{2}\delta _{0,t}$ is not in $\mathfrak{G}%
_{0}\left( A\right) $ and hence it does not correpond to any distribution.

\section{Stochastic differential equations}

\subsection{Grid differential equations\label{quello}}

A grid ordinary differential equation is a differential equation whose time
step ranges on the hyperfinite grid. This fact makes it to work as a discrete time
object simplifying many formal aspects.

A grid ordinary differential equation is then an equation of the kind
\begin{equation}
\frac{\Delta x}{\Delta t}(t)=f(t,x(t)),  \label{ode}
\end{equation}%
where $t\in {\mathbb{H}}$ , $x(t)$ is a grid function and $f:\mathbb{H}%
\times \mathbb{R}^{\ast }\rightarrow \mathbb{R}^{\ast }$ is an internal
function. A grid function $x(t)$ is a solution of the grid equation if
satisfies it at each point of the grid.

The following result shows that such an equation has an unique solution.
This, without regularity assumptions on the $f.$ Hence, this kind of
equations has solutions even if the equations contain a noise term (see
section (\ref{fopla})).

\begin{theorem}
Given an initial time $t_{0}\in \mathbb{H}$ and an initial data $x_{0}\in
\mathbb{R}^{\ast }$, the Cauchy problem associated to (\ref{ode}), that is
\begin{equation}
\left\{
\begin{array}{ll}
\frac{\Delta x}{\Delta t}(t)=f(t,x(t)) & t\in \mathbb{H} \\
x(t_{0})=x_{0} &
\end{array}%
\right.  \label{PC}
\end{equation}%
admits for $t\geq t_{0}$ an unique solution $x:\mathbb{H}\rightarrow \mathbb{%
R}^{\ast }$.
\end{theorem}

\begin{proof}
We know that $f$ is an ideal value of a sequence $\{f_{n}\}_{n\in \mathbb{N}%
} $. Also, we have that $t_{0}=t_{0,\alpha },\ x_{0}=x_{0,\alpha }$ are the ideal
values associated to $\{t_{0,n}\}_{n\in \mathbb{N}},\{x_{0,n}\}_{n\in
\mathbb{N}}$. For each $n\in \mathbb{N\ }$and $m\in \mathbb{Z},$ we can
construct by induction a sequence of functions.
\begin{equation*}
x_{n}:\frac{1}{n}\mathbb{Z\rightarrow R}
\end{equation*}%
as follows:%
\begin{eqnarray}
x_{n}\left( t_{0,n}\right) &=&x_{0,n}  \label{eqdiffin} \\
x_{n}\left( t_{0,n}+\frac{m+1}{n}\right) &=&x_{n}\left( t_{0,n}+\frac{m}{n}%
\right) + \\
&&+\frac{1}{n}f_{n}\left( \left( t_{0,n}+\frac{m}{n}\right) ,x_{n}\left(
t_{0,n}+\frac{m}{n}\right) \right) .
\end{eqnarray}%
Then by definition of internal function we have that, for $%
t=t_{0}+m\cdot\bigcirc\hspace{-6.5pt}\varepsilon\ ,\ m\in \mathbb{Z}^{\ast }$%
\begin{eqnarray}
x_{\alpha }\left( t_{0}\right) &=&x_{0} \\
x_{\alpha }\left( t+\bigcirc\hspace{-6.5pt}\varepsilon\ \right) &=&x_{\alpha
}\left( t\right) +\bigcirc\hspace{-6.5pt}\varepsilon\ f_{\alpha }\left(
t,x_{\alpha }\left( t\right) \right) .
\end{eqnarray}

Thus $x=x_{\alpha }$ solves (\ref{PC}). It is easy to check that this
solution is also unique.
\end{proof}

Given $x_0\sim x_1$, it may happen that
\begin{equation*}
x(t,x_0)\nsim x(t,x_1)
\end{equation*}
where $x(t,x_i)$ is the solution of (\ref{PC}) with initial data $x_i$. Some
times we would like to have
\begin{equation}  \label{solunica}
x(t,x_0)\sim x(t,x_1)\ \ \ \forall x_1\sim x_0;
\end{equation}
this can be useful, for example, when we want to consider the standard part
of a hyperreal differential equation.

We have the following proposition :

\begin{proposition}
Consider the following Cauchy problem
\begin{equation}
\left\{
\begin{array}{ll}
\frac{\Delta x}{\Delta t}(t)=f(t,x(t)) & t\in \mathbb{H}; \\
x(t_{0})=x_{0}, &
\end{array}%
\right.  \label{PC-lip}
\end{equation}%
and suppose that, there exists $L$ s.t.
\begin{equation}
|f(t,x)-f(t,y)|\leq L|x-y|.  \label{loclip}
\end{equation}

Let $x_{1}$ be a bounded initial data for the problem (\ref{PC-lip}) and $%
x(t,x_{1})$ is the solution of this problem. Then , if $x_{1}\sim x_{0}$ ,
then for all $0\leq t<T$ we have
\begin{equation}
x(t,x_{0})\sim x(t,x_{1})
\end{equation}
\end{proposition}

\begin{proof}
Arguing as in standard analysis, we have that (\ref{loclip}) guarantees
that, for any $T_{1}<T,$the solution is bounded. Moreover, in standard
analysis, the condition (\ref{loclip}) guarantees, the continuous dependence
of the solution from initial data $x_{0}$. In our case, until the solution
is finite, we can proceed in the same way to prove that, chosen an arbitrary
$T_{1}<T$
\begin{equation}
|x(t,x_{0})-x(t,x_{1})|\leq |x_{0}-x_{1}|e^{LT_{1}}
\end{equation}%
for all $t\in \lbrack 0,T_{1}]$. Because $x_{0}\sim x_{1}$ by hypothesis, we
have that
\begin{equation}
x(t,x_{0})\sim x(t,x_{1})
\end{equation}%
for all $t\in \lbrack 0,T_{1}]$. This assures the proof.
\end{proof}

\subsection{Stochastic grid equations and the Fokker-Plank equation\label%
{fopla}}

In our approach, a stochastic differential equation consists of a set of
grid differential equations. Each differential equation has a \textit{noise}
term and gives a trajectory which can be considered as a realization of a
process.

Let $\mathcal{R}\subset \mathfrak{G}\left[ 0,1\right] $ be an hyperfinite
set of grid functions and consider the class of Cauchy problems%
\begin{equation}
\left\{
\begin{array}{l}
\frac{\Delta x}{\Delta t}(t)=f(t,x)+h(t,x)\xi , \\
x(0)=x_{0}, \\
\xi (t)\in \mathcal{R}.%
\end{array}%
\right.  \label{gina}
\end{equation}%
where%
\begin{equation*}
f,h:\left[ 0,1\right] _{\mathbb{H}}\times \mathbb{R}^{\ast }\rightarrow
\mathbb{R}^{\ast }
\end{equation*}

We want to study the statistical behavior of the set of solutions of the
above Cauchy problems
\begin{equation*}
\mathcal{S}=\left\{ x_{\xi }(t):\xi \in \mathcal{R}\right\} ;
\end{equation*}%
More precisely we want to describe the behavior of the density function
\begin{equation*}
\rho :\left[ 0,1\right] _{\mathbb{H}}\times \mathbb{H}\rightarrow \mathbb{Q}%
^{\ast }
\end{equation*}%
defined as follows%
\begin{equation*}
\rho \left( t,x\right) =\frac{\left\vert \{x_{\xi }\in \mathcal{S}:x\leq
x_{\xi }(t)<x+\bigcirc\hspace{-6.5pt}\varepsilon\ \}\right\vert }{\bigcirc%
\hspace{-6.5pt}\varepsilon\ \left\vert \mathcal{R}\right\vert }.
\end{equation*}

We are interested in the case in which $\mathcal{R}$ models a \textit{white
noise}; roughly speaking we can define a white noise as the hyperfinite set
of all the grid functions with values $\pm \sqrt{\alpha }.$ Here there is
its precise definition:

\begin{definition}
The white noise is the set of grid functions defined by
\begin{equation*}
\mathcal{R}=\mathcal{R}_{\alpha }
\end{equation*}%
where%
\begin{equation*}
\mathcal{R}_{n}=\left\{ -\sqrt{n},+\sqrt{n}\right\} ^{\left[ 0,1\right] _{%
\mathbb{H}_{n}}}
\end{equation*}
\end{definition}

Hence $\mathcal{R}$ is a hyperfinite set with $\left\vert \mathcal{R}%
\right\vert =2^{\alpha +1}.$

\bigskip

\begin{remark}
We would be tempted to write%
\begin{equation*}
\mathcal{R}_{\alpha }=\left\{ -\sqrt{\alpha },+\sqrt{\alpha }\right\} ^{%
\left[ 0,1\right] _{\mathbb{H}_{\alpha }}}
\end{equation*}%
however this notation is very ambiguous; in fact $\mathcal{R}_{\alpha }$ is
a set defined by the Internal Set Axiom and it contains only \emph{internal}
function. However the symbol $\left\{ -\sqrt{\alpha },+\sqrt{\alpha }%
\right\} ^{\left[ 0,1\right] _{\mathbb{H}_{\alpha }}}$ usually represents
the set of \emph{all} the functions $f:\left[ 0,1\right] _{\mathbb{H}%
_{\alpha }}\rightarrow \left\{ -\sqrt{\alpha },+\sqrt{\alpha }\right\} .$
\end{remark}

Now, we can state the main result of this paper:

\begin{theorem}
\label{fokk}Assume that $\mathcal{R}$ is a white noise and that $f(t,x)$ and
$h(t,x)$ are continuous functions. Then the distribution $T_{\rho }$
relative to the density function $\rho $ is a measure and satisfies the
\emph{Fokker-Plank equation}%
\begin{equation}
\frac{dT_{\rho }}{dt}+\frac{d}{dx}\left( f(t,x)T_{\rho }\right) -\frac{1}{2}%
\frac{d^{2}}{dx^{2}}\left( h(t,x)^{2}T_{\rho }\right) =0.  \label{fp}
\end{equation}%
\begin{equation}
T_{\rho }(0,x)=\delta  \label{ic}
\end{equation}%
\bigskip in the sense of distribution.
\end{theorem}

\begin{remark}
We recall that (\ref{fp}) and (\ref{ic}) "in the sense of distributions"
mean that $T_{\rho }$ satisfies the equation%
\begin{equation}
\left\langle \varphi _{t}+f\varphi _{x}+\frac{1}{2}h^{2}\varphi
_{xx},T_{\rho }\right\rangle +\varphi \left( x_{0}\right) =0  \label{dis}
\end{equation}%
for any $\varphi \in \mathcal{D}\left( \left[ 0,1\right) \times \mathbb{R}%
\right) .$ The duality $\left\langle \cdot ,\cdot \right\rangle $ is between
the space of continuous function and the space of measures. Equation (\ref%
{dis}) can be expressed using the grid function $\rho $ and the $\alpha $%
-integral by the following equation:%
\begin{equation}
\forall \varphi \in \mathcal{D}\left( \left[ 0,1\right) \times \mathbb{R}%
\right) ,\ \iint \left( \varphi _{t}+f\varphi _{x}+\varphi _{xx}h^{2}\right)
\rho \ \Delta x\ \Delta t+\varphi (0,x_{0})=0  \label{jessica}
\end{equation}%
Actually, we will prove Th. \ref{fokk} just proving the above equation.
\end{remark}

\begin{remark}
If $f(t,x)$ and $h(t,x)$ are smooth functions, by standard results in PDE,
we know that, for $t>0,$ the distribution $T_{\rho }$ coincides with a
smooth function $u(t,x)$. Then, for any $t>0,$ $\rho $ defines a smooth
function $u$ by the formula%
\begin{equation*}
\forall \varphi \in \mathcal{D}\left( \left( 0,1\right) \times \mathbb{R}%
\right) ,\ \iint \rho \varphi \ \Delta x \Delta t= \iint u\varphi \ dx\ dt
\end{equation*}%
and $u$ satisfies the Fokker-Plank equation in $\left( 0,1\right) \times
\mathbb{R}$ in the usual sense.
\end{remark}

\begin{remark}
We will see in the proof of Th. (\ref{fokk}) that if the functions $f(t,x)$
and $h(t,x)$ are not continuous, but only bounded on compact sets, the
equation (\ref{jessica}) still holds. However in this case, equation (\ref%
{jessica}) cannot be interpreted so easily. For example, if $f(t,x)$ and $%
h(t,x)$ are not measurable, there is no simple standard interpretation.
\end{remark}

\bigskip

Given $t\in \left[ 0,1\right] _{\mathbb{H}}$ we set
\begin{equation}
\mathcal{R}[0,t)=\mathcal{R}_{\alpha }[0,t);\ \ \mathcal{R}_{n}:=\left\{ -%
\sqrt{n},+\sqrt{n}\right\} ^{[0,t)_{\mathbb{H}_{n}}};  \label{rb}
\end{equation}%
namely, $\mathcal{R}[0,t)$ is the set of the restrictions of the functions
of $\mathcal{R}$ to $[0,t)_{\mathbb{H}}.$ Moreover, for $\tau \in \mathcal{R}%
\left[ 0,s\right) ,$ we set
\begin{equation*}
\mathcal{R}_{\tau }\left[ s,1\right] =\left\{ \xi \in \mathcal{R}:\xi \left(
t\right) =\tau \left( t\right) \text{ for }t<s\right\}
\end{equation*}%
So we have the following decomposition:%
\begin{equation}
\mathcal{R}=\bigcup_{\tau \in \mathcal{R}\left[ 0,s\right) }\mathcal{R}%
_{\tau }\left[ s,1\right] .  \label{decom}
\end{equation}%
We define the mean value of a grid function in the set $\left[ x,y\right]
\cap {\mathbb{H}}$ as follows:
\begin{equation*}
{\mathbb{E}}_{\left[ x,y\right) }\left[ f\right] =\frac{1}{(y-x)}\mathbb{I}_{%
\left[ x,y\right] }\left[ f\right] =\frac{\bigcirc\hspace{-6.5pt}%
\varepsilon\ }{(y-x)}\sum_{t\in \left[ x,y\right) \cap {\mathbb{H}}}f\left(
t\right)
\end{equation*}

In general, if $\Gamma $ is a hyperfinite set and $\Phi :\Gamma \rightarrow
\mathbb{R}^{\ast }$ is an internal function, the mean value of $\Phi $ in $%
\Gamma $ is defined as follows:%
\begin{equation*}
{\mathbb{E}}_{\xi \in \Gamma }\left[ \Phi \right] =\frac{1}{\left\vert
\Gamma \right\vert }\sum_{\xi \in \Gamma }\Phi \left( \xi \right)
\end{equation*}

\begin{proposition}
\label{classofnoises} If $\mathcal{R}$ is a white noise, then for any $t\in %
\left[ 0,1\right] _{\mathbb{H}},$ and $\tau \in \mathcal{R}\left[ 0,t\right]
,$ we have
\begin{equation}
\text{the hyperfinite number }\left\vert \mathcal{R}_{\tau }\left[ t,1\right]
\right\vert \ \text{does not depend on }\tau \in \mathcal{R}\left[ 0,t\right]
\label{ind}
\end{equation}%
and
\begin{eqnarray}
\mathbb{E}_{\xi \in \mathcal{R}_{\tau }\left[ t,1\right] }\left[ \xi (t)%
\right] &\sim &0  \label{med} \\
\mathbb{E}_{\xi \in \mathcal{R}_{\tau }\left[ t,1\right] }\left[ \xi (t)^{2}%
\right] &\sim &\alpha ,  \label{var}
\end{eqnarray}
\end{proposition}

\begin{proof}
The proof is almost immediate: first of all we have that%
\begin{equation*}
\left\vert \mathcal{R}_{\tau }\left[ t,1\right] \right\vert =2^{\alpha
(1-t)+1};
\end{equation*}%
moreover%
\begin{eqnarray*}
\mathbb{E}_{\xi \in \mathcal{R}_{\tau }\left[ t,1\right] }\left[ \xi (t)%
\right] &=&\frac{1}{\left\vert \mathcal{R}_{\tau }\left[ t,1\right]
\right\vert }\dsum\limits_{\xi \in \mathcal{R}_{\tau }\left[ t,1\right] }\xi
(t) \\
&=&\frac{1}{2\left\vert \mathcal{R}_{\tau }\left[ t,1\right] \right\vert }%
\sqrt{\alpha }-\frac{1}{2\left\vert \mathcal{R}_{\tau }\left[ t,1\right]
\right\vert }\sqrt{\alpha }=0\sim 0
\end{eqnarray*}%
and%
\begin{equation*}
\mathbb{E}_{\xi \in \mathcal{R}_{\tau }\left[ t,1\right] }\left[ \xi (t)^{2}%
\right] =\frac{1}{\left\vert \mathcal{R}_{\tau }\left[ t,1\right]
\right\vert }\dsum\limits_{\xi \in \mathcal{R}_{\tau }\left[ t,1\right] }\xi
(t)^{2}=\alpha
\end{equation*}
\end{proof}

\bigskip

\begin{remark}
The conclusion of Th. \ref{fokk} hold not only if the "stochastic class" $%
\mathcal{R}$ is defined by (\ref{rb}), but for any class $\mathcal{R}$ which
satisfies the properties (\ref{ind}), (\ref{med}) and (\ref{var}). For
example we can take%
\begin{equation*}
\mathcal{R=R}_{\alpha };\ \ \mathcal{R}_{n}:=\left\{ q_{1}\sqrt{n},....,q_{k}%
\sqrt{n}\right\} ^{\left[ 0,1\right] _{\mathbb{H}_{n}}};\ \ k\in \mathbb{N}
\end{equation*}%
with $q_{i}\in \mathbb{R}^{\ast },$%
\begin{equation*}
\dsum\limits_{i=1}^{k}q_{i}=0;\ \dsum\limits_{i=1}^{k}q_{i}^{2}=1.
\end{equation*}
\end{remark}

The following two lemmas are a direct consequence of properties (\ref{ind}),
(\ref{med}) and (\ref{var}).

\begin{lemma}
\label{media} Let $G:\left[ 0,1\right] _{\mathbb{H}}\times \mathbb{R}^{\ast
}\times \mathbb{R}^{\ast }\rightarrow \mathbb{R}^{\ast }\mathbb{\ }$be any
internal function. Then, for every $t\in \left[ \bigcirc\hspace{-6.5pt}%
\varepsilon\ ,1\right] $
\begin{equation*}
\mathbb{E}_{\xi \in \mathcal{R}}\left[ G(t,x_{\xi }\left( t\right) ,\xi
\left( t\right) )\right] =\mathbb{E}_{\tau \in \mathcal{R}\left[ 0,t\right) }%
\mathbb{E}_{\xi \in \mathcal{R}_{\tau }\left[ t,1\right] }\left[ G(t,x_{\xi
}\left( t\right) ,\xi \left( t\right) )\right]
\end{equation*}
\end{lemma}

\begin{proof}
By (\ref{decom}), we have that$\ $%
\begin{equation*}
\mathcal{R=}\bigcup_{\tau \in \mathcal{R}\left[ 0,t\right) }\mathcal{R}%
_{\tau }\left[ t,1\right]
\end{equation*}%
Then,%
\begin{eqnarray*}
\mathbb{E}_{\xi \in \mathcal{R}}\left[ G\left( t,x_{\xi }\left( t\right)
,\xi \left( t\right) \right) \right] &=&\frac{1}{\left\vert \mathcal{R}%
_{\tau }\left[ t,1\right] \right\vert \cdot \left\vert \mathcal{R}\left[
0,t\right) \right\vert }\sum_{\tau \in \mathcal{R}\left[ 0,t\right)
}\sum_{\xi \in \mathcal{R}_{\tau }\left[ t,1\right] }G\left( t,x_{\xi
}\left( t\right) ,\xi \left( t\right) \right) \\
&=&\frac{1}{\left\vert \mathcal{R}\left[ 0,t\right) \right\vert }\sum_{\tau
\in \mathcal{R}\left[ 0,t\right] }\frac{1}{\left\vert \mathcal{R}_{\tau }%
\left[ t,1\right] \right\vert }\sum_{\xi \in \mathcal{R}_{\tau }\left[ t,1%
\right] }G\left( t,x_{\xi }\left( t\right) ,\xi \left( t\right) \right) \\
&=&\mathbb{E}_{\tau \in \mathcal{R}\left[ 0,t\right) }\mathbb{E}_{\xi \in
\mathcal{R}_{\tau }\left[ t,1\right] }\left[ G(t,x_{\xi }\left( t\right)
,\xi \left( t\right) )\right] .
\end{eqnarray*}
\end{proof}

\begin{lemma}
\label{ganzo} Let $F:\left[ 0,1\right] _{\mathbb{H}}\times \mathbb{R}^{\ast
}\rightarrow \mathbb{R}^{\ast }\mathbb{\ }$be an internal function such that
$|F(t,x)|\leq M,\ M\in \mathbb{R}$. Then, for every $t\in \left[ 0,1\right] $
\begin{eqnarray*}
\mathbb{E}_{\xi \in \mathcal{R}}\left[ F(t,x_{\xi }(t))\cdot \xi (t)\right]
&\sim &0. \\
\mathbb{E}_{\xi \in \mathcal{R}}\left[ F(t,x_{\xi }(t))\cdot \xi (t)^{2}%
\right] &\sim &\alpha \cdot \mathbb{E}_{\xi \in \mathcal{R}}\left[
F(t,x_{\xi }(t))\right]
\end{eqnarray*}
\end{lemma}

\begin{proof}
By lemma \ref{media}, we have that$\ $%
\begin{equation*}
\mathbb{E}_{\xi \in \mathcal{R}}\left[ F(t,x_{\xi }(t))\cdot \xi (t)\right] =%
\mathbb{E}_{\tau \in \mathcal{R}\left[ 0,t\right) }\mathbb{E}_{\xi \in
\mathcal{R}_{\tau }\left[ t,1\right] }\left[ F(t,x_{\xi }(t))\cdot \xi (t)%
\right]
\end{equation*}%
Since $x_{\xi }(t)$ does not depend on $\xi (s)$ for $s>t,$ we have that%
\begin{eqnarray*}
\mathbb{E}_{\xi \in \mathcal{R}_{\tau }\left[ t,1\right] }\left[ F(t,x_{\xi
}(t))\cdot \xi (t)\right] &=&\frac{1}{\left\vert \mathcal{R}_{\tau }\left[
t,1\right] \right\vert }\sum_{\tau \in \mathcal{R}\left[ 0,t\right] }\left(
F(t,x_{\xi }(t))\cdot \xi (t)\right) \\
&=&F(t,x_{\xi }(t))\cdot \frac{1}{\left\vert \mathcal{R}_{\tau }\left[ t,1%
\right] \right\vert }\sum_{\tau \in \mathcal{R}\left[ 0,t\right] }\xi (t) \\
&=&F(t,x_{\xi }(t))\cdot \mathbb{E}_{\xi \in \mathcal{R}_{\tau }\left[ t,1%
\right] }\left[ \xi (t)\right]
\end{eqnarray*}%
Then since $F$ is bounded, by (\ref{med}), we get the conclusion:%
\begin{equation*}
\mathbb{E}_{\xi \in \mathcal{R}}\left[ F(t,x_{\xi }(t))\xi (t)\right] =%
\mathbb{E}_{\tau \in \mathcal{R}\left[ 0,t\right) }\left[ F(t,x_{\xi
}(t))\cdot \mathbb{E}_{\xi \in \mathcal{R}_{\tau }\left[ t,1\right] }\left[
\xi (t)\right] \right] \sim 0
\end{equation*}

Analogously, we have that
\begin{equation*}
\mathbb{E}_{\xi \in \mathcal{R}}\left[ F(t,x_{\xi }(t))\xi (t)^{2}\right] =%
\mathbb{E}_{\tau \in \mathcal{R}\left[ 0,t\right) }\left[ F(t,x_{\xi }(t))%
\mathbb{E}_{\xi \in \mathcal{R}_{\tau }\left[ t,1\right] }\left[ \xi (t)^{2}%
\right] \right]
\end{equation*}%
and by (\ref{var}), we get that%
\begin{equation*}
F(t,x_{\xi }(t))\mathbb{E}_{\xi \in \mathcal{R}_{\tau }\left[ t,1\right] }%
\left[ \xi (t)^{2}\right] =F(t,x_{\xi }(t))\left( \alpha +\varepsilon _{\tau
}\right)
\end{equation*}%
where $\varepsilon _{\tau }\sim 0.$ Then%
\begin{eqnarray*}
\mathbb{E}_{\xi \in \mathcal{R}}\left[ F(t,x_{\xi }(t))\xi (t)^{2}\right]
&=&\alpha \mathbb{E}_{\tau \in \mathcal{R}\left[ 0,t\right) }\left[
F(t,x_{\xi }(t))\right] +\mathbb{E}_{\tau \in \mathcal{R}\left[ 0,t\right) }%
\left[ F(t,x_{\xi }(t))\varepsilon _{\tau }\right] \\
&\sim &\alpha \mathbb{E}_{\tau \in \mathcal{R}\left[ 0,t\right) }\left[
F(t,x_{\xi }(t))\right]
\end{eqnarray*}

This concludes the proof.
\end{proof}

\bigskip

Now we see a basic property of the density function.

\begin{lemma}
\label{den}Let $\varphi \in \mathcal{D}\left( \left[ 0,1\right] \times
\mathbb{R}\right) $ and let $x_{\xi }(t),$ $\xi \in \mathcal{R},$ be the
family of solutions of a grid stochastic ODE. Then
\begin{equation*}
\mathbb{E}_{\xi \in \mathcal{R}}\left[ \varphi (t,x_{\xi }(t))\right] \sim
\bigcirc\hspace{-6.5pt}\varepsilon\ \sum_{x\in \mathbb{H}}\varphi (t,x)\rho
(t,x).
\end{equation*}%
In particular,
\begin{equation*}
\mathbb{E}_{\xi \in \mathcal{R}}\left[ \varphi (t,x_{\xi }(t))\right] \sim
\int \varphi (t,x)\rho (t,x)\Delta x.
\end{equation*}
\end{lemma}

\begin{proof}
We have
\begin{eqnarray*}
\mathbb{E}_{\xi \in \mathcal{R}}\left[ \varphi (t,x_{\xi }(t))\right] &=&%
\frac{1}{\left\vert \mathcal{R}\right\vert }\sum_{\xi \in \mathcal{R}%
}\varphi (t,x_{\xi }(t))=\frac{1}{\left\vert \mathcal{R}\right\vert }%
\sum_{x\in \mathbb{H}}\left[ \sum_{x\leq x_{\xi }(t)<x+\bigcirc\hspace{-5.5pt%
}\varepsilon\ }\varphi (t,x_{\xi }(t))\right] \\
&=&\frac{1}{\left\vert \mathcal{R}\right\vert }\sum_{\xi \in \mathcal{R}%
}\varphi (t,x_{\xi }(t))=\frac{1}{\left\vert \mathcal{R}\right\vert }%
\sum_{x\in \mathbb{H}}\left[ \sum_{x\leq x_{\xi }(t)<x+\bigcirc\hspace{-5.5pt%
}\varepsilon\ }\left( \varphi (t,x)+\eta _{\xi }(t)\right) \right]
\end{eqnarray*}%
where $\eta _{\xi }(t):=\varphi (t,x_{\xi }(t))-\varphi (t,x)$ is
infinitesimal. Hence%
\begin{equation*}
\frac{1}{\left\vert \mathcal{R}\right\vert }\sum_{\xi \in \mathcal{R}}\eta
_{\xi }(t)\sim 0
\end{equation*}%
and so%
\begin{eqnarray*}
\mathbb{E}_{\xi \in \mathcal{R}}\left[ \varphi (t,x_{\xi }(t))\right] &=&%
\frac{1}{\left\vert \mathcal{R}\right\vert }\sum_{x\in \mathbb{H}}\left[
\sum_{x\leq x_{\xi }(t)<x+\bigcirc\hspace{-5.5pt}\varepsilon\ }\varphi (t,x)%
\right] \\
&=&\frac{1}{\left\vert \mathcal{R}\right\vert }\sum_{x\in \mathbb{H}}\left(
\varphi (t,x)\cdot \left\vert \{x_{\xi }\in \mathcal{S}:x\leq x_{\xi
}(t)<x+\bigcirc\hspace{-6.5pt}\varepsilon\ \right\vert \right) \\
&=&\frac{1}{\left\vert \mathcal{R}\right\vert }\sum_{x\in \mathbb{H}}\left[
\varphi (t,x)\cdot \bigcirc\hspace{-6.5pt}\varepsilon\ \left\vert \mathcal{R}%
\right\vert \rho (t,x)\right] =\bigcirc\hspace{-6.5pt}\varepsilon\
\sum_{x\in \mathbb{H}}\left[ \varphi (t,x)\rho (t,x)\right]
\end{eqnarray*}
\end{proof}

Now we can prove Theorem \ref{fokk}

\begin{proof}
Chosen an arbitrary $\varphi \in \mathcal{C}(\left[ 0,1\right] \times
\mathbb{R})$ bounded in the second variable, we have that
\begin{equation*}
\varphi (1,x_{\xi }(1))-\varphi (0,x_{0})=\bigcirc\hspace{-6.5pt}%
\varepsilon\ \sum_{t\in \left[ 0,1-\bigcirc\hspace{-6.5pt}\varepsilon\ %
\right] _{\mathbb{H}}{}}\frac{\Delta \varphi }{\Delta t}(t,x_{\xi }(t)),
\end{equation*}%
Now we assume that $\varphi \in \mathcal{D}(\left[ 0,1\right) \times \mathbb{%
R})).$ Since $x_{\xi }$ solves eq. (\ref{gina}), we can apply Theorem \ref%
{ITO}, and we obtain
\begin{eqnarray}
-\varphi (0,x_{0}) &\sim &\bigcirc\hspace{-6.5pt}\varepsilon\ \sum_{t\in %
\left[ 0,1\right) _{\mathbb{H}}{}}\left[ \varphi _{t}+\varphi _{x}\cdot
\frac{\Delta x}{\Delta t}+\frac{\bigcirc\hspace{-6.5pt}\varepsilon\ }{2}%
\varphi _{xx}\cdot \left( \frac{\Delta x}{\Delta t}\right) ^{2}\right] =
\notag \\
&\sim &\bigcirc\hspace{-6.5pt}\varepsilon\ \sum_{t\in \left[ 0,1\right) _{%
\mathbb{H}}{}}\left[ \varphi _{t}+\varphi _{x}\cdot \left( f+h\xi \right) +%
\frac{\bigcirc\hspace{-6.5pt}\varepsilon\ }{2}\varphi _{xx}\cdot \left(
f+h\xi \right) ^{2}\right] =  \label{70} \\
&=&\bigcirc\hspace{-6.5pt}\varepsilon\ \sum_{t\in \left[ 0,1\right) _{%
\mathbb{H}}{}}\left( \varphi _{t}+f\varphi _{x}\right) +\left( \varphi
_{x}h+\bigcirc\hspace{-6.5pt}\varepsilon\ \varphi _{xx}f\right) \xi +\frac{%
\bigcirc\hspace{-6.5pt}\varepsilon\ }{2}\varphi _{xx}f+\frac{\bigcirc\hspace{%
-6.5pt}\varepsilon\ }{2}\varphi _{xx}h^{2}\xi ^{2}  \notag
\end{eqnarray}

Now we want to compute $\mathbb{E}_{\xi \in \mathcal{R}}$ of each piece of
the right hand side of the above equation: by lemma \ref{den},%
\begin{equation}
\mathbb{E}_{\xi \in \mathcal{R}}\left[ \varphi _{t}+f\varphi _{x}\right]
=\bigcirc\hspace{-6.5pt}\varepsilon\ \sum_{x\in \mathbb{H}}\left[ \varphi
_{t}+f\varphi _{x}\right] \rho  \label{71}
\end{equation}%
for every $t\in \left[ 0,1\right] .$

Let us consider the second piece: $\left( \varphi _{x}h+\bigcirc\hspace{%
-6.5pt}\varepsilon\ \varphi _{xx}f\right) \xi .$ If we set
\begin{equation*}
F(t,x_{\xi }(t))=\varphi _{x}(t,x_{\xi }(t))h(t,x_{\xi }(t))+\bigcirc\hspace{%
-6.5pt}\varepsilon\ \varphi _{xx}(t,x_{\xi }(t))f(t,x_{\xi }(t))
\end{equation*}%
it turns out that $F(t,x_{\xi }(t))$ is bounded. In fact if $x_{\xi }(t)$ is
bounded, $\varphi _{x}(t,x_{\xi }(t))$, $\ h(t,x_{\xi }(t))$, $\varphi
_{xx}(t,x_{\xi }(t))$ and $f(t,x_{\xi }(t))$ are bounded since they are
standard functions; if $x_{\xi }(t)$ is unbounded, $\varphi _{x}(t,x_{\xi
}(t))=0.$ Then we can apply lemma \ref{ganzo} and we get that
\begin{equation}
\mathbb{E}_{\xi }\left[ \left( \varphi _{x}h+\bigcirc\hspace{-6.5pt}%
\varepsilon\ \varphi _{xx}f\right) \xi \right] \sim 0.  \label{73}
\end{equation}%
Moreover
\begin{equation}
\mathbb{E}_{\xi }\left[ \frac{\bigcirc\hspace{-6.5pt}\varepsilon\ }{2}%
\varphi _{xx}f\right] \sim 0  \label{74}
\end{equation}%
since $\frac{\bigcirc\hspace{-5.5pt}\varepsilon\ }{2}\varphi _{xx}f\sim 0$.

Finally, let us see the last term:$\ \frac{\bigcirc\hspace{-5.5pt}%
\varepsilon\ }{2}\varphi _{xx}h^{2}\xi ^{2}.$ We can see that $\frac{\bigcirc%
\hspace{-5.5pt}\varepsilon\ }{2}\varphi _{xx}h^{2}$ is bounded (actually, it
is infinitesimal); then we can apply lemma \ref{ganzo}\ with $F(t,x_{\xi
}(t))=\frac{\bigcirc\hspace{-5.5pt}\varepsilon\ }{2}\varphi _{xx}(t,x_{\xi
}(t))h(t,x_{\xi }(t))^{2}$:

\begin{eqnarray}
\mathbb{E}_{\xi }\left[ \frac{\bigcirc\hspace{-6.5pt}\varepsilon\ }{2}%
\varphi _{xx}h^{2}\xi ^{2}\right] &\sim &\alpha \mathbb{E}_{\xi }\left[
\frac{\bigcirc\hspace{-6.5pt}\varepsilon\ }{2}\varphi _{xx}h^{2}\right]
\notag \\
&=&\frac{1}{2}\mathbb{E}_{\xi }\left[ \varphi _{xx}h^{2}\right]  \notag \\
&=&\frac{1}{2\left\vert \mathcal{R}\right\vert }\sum_{\xi \in \mathcal{R}%
}\varphi _{xx}(t,x\left( t\right) )h(t,x\left( t\right) )^{2}  \label{75}
\end{eqnarray}%
Then, by lemma \ref{den}%
\begin{equation*}
\mathbb{E}_{\xi }\left[ \frac{\bigcirc\hspace{-6.5pt}\varepsilon\ }{2}%
\varphi _{xx}h^{2}\xi ^{2}\right] \sim \frac{\bigcirc\hspace{-6.5pt}%
\varepsilon\ }{2}\sum_{x\in \mathbb{H}}\varphi _{xx}(t,x)h(t,x)^{2}\rho (t,x)
\end{equation*}

Then, by (\ref{70}),...,(\ref{75})
\begin{eqnarray*}
-\varphi (0,x_{0}) &=&\mathbb{E}_{t,\xi }\left[ -\varphi (0,x_{0})\right] \\
&\sim &\bigcirc\hspace{-6.5pt}\varepsilon\ \sum_{t\in \left[ 0,1\right) _{%
\mathbb{H}}{}}\Big(\mathbb{E}_{\xi }\left[ \varphi _{t}+f\varphi _{x}\right]
+\mathbb{E}_{\xi }\left[ \left( \varphi _{x}h+\bigcirc\hspace{-6.5pt}%
\varepsilon\ \varphi _{xx}f\right) \xi \right] \Big)+ \\
&&+\bigcirc\hspace{-8.5pt}\varepsilon\ \sum_{t\in \left[ 0,1\right) _{%
\mathbb{H}}{}}\Big(\mathbb{E}_{\xi }\left[ \frac{\bigcirc\hspace{-6.5pt}%
\varepsilon\ }{2}\varphi _{xx}f\right] +\mathbb{E}_{\xi }\left[ \frac{%
\bigcirc\hspace{-6.5pt}\varepsilon\ }{2}\varphi _{xx}h^{2}\xi ^{2}\right] %
\Big) \\
&\sim &\bigcirc\hspace{-6.5pt}\varepsilon\ ^{2}\sum_{t\in \left[ 0,1\right)
_{\mathbb{H}}{}}\left( \sum_{x\in \mathbb{H}}\left( \varphi _{t}+f\varphi
_{x}\right) \rho +\varphi _{xx}h^{2}\rho \right) \\
&\sim &\int \int \left( \varphi _{t}+f\varphi _{x}+\varphi _{xx}h^{2}\right)
\rho \ dxdt
\end{eqnarray*}

Since the first and the last term of this equation are standard, we get eq. (%
\ref{jessica}).
\end{proof}

\bigskip

\section{Conclusion}

\bigskip

In this section we will make some comments suggested by the results of this
paper.

\bigskip

The first comment is that,

\begin{itemize}
\item \emph{in many applications of Nonstandard Analysis, only elementary
facts and techniques of nonstandard calculus seem to be necessary. }
\end{itemize}

In fact $\alpha $-theory is much simpler than the usual Nonstandard
Analysis, but it seems absolutely adequate to treat the kind of problems
considered here.

However, if you compare this paper with \cite{and}, \cite{kei} and \cite{nel}%
, the reason why this paper is much simpler lies in the fact that we have
not made a nonstandard theory of the \textit{stochastic differential
equation, }but rather, we have replaced them with the \textit{stochastic
grid equations }which are much simpler mathematical objects; in this
contest the Ito integral is replaced by the $\alpha $-integral and the proof
of the key point of this theory, the Ito formula, reduces to an exercise.

The reason why a much simpler object can represent the processes of
diffusion at microscopic level is that we have taken the infinitesimals
seriously and we have used them to model an aspect of the "physical
reality". We think that in general,

\begin{itemize}
\item \emph{the advantages of a theory which includes infinitesimals rely
more on the possibility of making new models rather than in the
demonstration techniques.}
\end{itemize}

In the case considered in this paper, the way the model has been constructed
appears quite natural: the \textit{stochastic grid equations,} which need
the notion of infinitesimals, describe the diffusion processes at
microscopic level; the Fokker-Plank equation describes the diffusion
processes at the macroscopic level and uses standard differential equations
(and the theory of distributions when the data are not regular).

The connection of these two levels is given by eq. (\ref{d}) which relates
grid function (microscopic level) with distributions (macroscopic level).

A final remark concerns the \emph{theory of probability}. We have not used
the notion of probability to show that every thing can be kept to a very
elementary level and that our variant of the Ito formula makes sense also in
a context where probability does not appear.

However probability can be introduced in a very elementary way. We may think
of the stochastic class $\mathcal{R}$ as a sample space. The events are the
hyperfinite sets $E\subset \mathcal{P}\left( \mathcal{R}\right) ^{\ast }$
and the probability $P$ of an event is given by%
\begin{equation}
P\left( E\right) =\frac{\left\vert E\right\vert }{\left\vert \mathcal{R}%
\right\vert }.  \label{pro}
\end{equation}

For example $E$ might be the event that at a time $t_{0},$ the moving
particle lies in the interval $\left[ a,b\right] ,$ namely%
\begin{equation*}
E=\left\{ x_{\xi }\in \mathcal{S}:x_{\xi }\left( t_{0}\right) \in \left[ a,b%
\right] _{\mathbb{H}}\right\}
\end{equation*}%
In this case we have that%
\begin{equation*}
P\left( E\right) \sim \int_{a}^{b}\rho (t,x)\ \Delta x
\end{equation*}

Obviously, this is the most natural extension to infinite sample spaces of
the classical definition of probability. There is no doubt that definition (%
\ref{pro}), is the most simple and intuitive definition of probability. The
price which you pay is that $P$ takes its values in $\mathbb{Q}^{\ast }$ and
not in $\mathbb{R};$ thus, a probability space $\left( \mathcal{R},P\right) $
is an internal object and, if it is infinite, it is not standard. The
problem arise if you want to connect the nonstandard world with the standard
one. This operation can be done in a very elegant way via the Loeb integral
(see \cite{loeb} or also \cite{kei}). A different way to make easy the
connection between this two worlds has been proposed by Nelson \cite{nel}.
However, we think that, in most cases, it is not necessary to make this
connection. Only at the very end you may consider the shadow of the numbers
which you have obtained. At least this can be done in the theory which we
have exposed in this paper. So, if we accept a mathematical description in
which the infinitesimal exist,

\begin{itemize}
\item \emph{the probability of an event is given by a hyperrational number,}
\end{itemize}

\noindent and many theorems become simpler. This scheme avoids also some
facts in probability theory which are in contrast with the common sense: for
example the fact that the union of impossible events gives a possible one
(the probability of a non-denumerable set might be non-null even if the
probability of each singleton is null).

\bigskip

Concluding the final remark of this paper could be the following one:

\begin{itemize}
\item \emph{the infinitesimals should be taken seriously.}
\end{itemize}

\end{document}